\documentclass[letterpaper, 10 pt, conference]{ieeeconf}  

\IEEEoverridecommandlockouts                              

\usepackage{balance}

\usepackage{amsmath} 
\usepackage{amssymb}  
\usepackage[lined,ruled,vlined]{algorithm2e}
\usepackage{bm}
\usepackage{srcltx}
\usepackage{cite}    
\usepackage{verbatim}
\usepackage{float}
\usepackage{acronym}
\usepackage{color}
\usepackage[lined,ruled,vlined]{algorithm2e}

\usepackage{paralist}
\usepackage[left=20mm, right=20mm, top=20mm,
  bottom=20mm]{geometry}

\usepackage{tikz}
\usetikzlibrary{shapes,snakes}

\usepackage{placeins}
\usetikzlibrary{arrows,fit,shapes,automata}
\usetikzlibrary{positioning,fit,calc,shapes}
\usetikzlibrary{decorations.fractals}
\usetikzlibrary{decorations.markings}

\usepackage{caption}
\usepackage{subcaption}

 \usepackage{acronym}
\newtheorem{definition}{Definition}

\newtheorem{lemma}{Lemma}

\newtheorem{problem}{Problem}

\newtheorem{assumption}{Assumption}

\acrodef{ltl}[LTL]{linear temporal logic}
\acrodef{pomdps}[POMDPs]{partially observable Markov decision processes}

\providecommand{\abs}[1]{\lvert#1\rvert}
\usepackage{mathtools}

\newcommand{\defeq}{\vcentcolon=}
\newcommand{\knows}{\mathbf{Knows}}

\newcommand{\obs}{\mathsf{Obs}}
\newcommand{\calAP}{\mathcal{AP}}
\newcommand{\calO}{\mathcal{O}}
\newcommand{\calB}{\mathcal{B}}
\newcommand{\calP}{\mathcal{P}}
\newcommand{\calD}{\mathcal{D}}

\newcommand{\calA}{\mathcal{A}}

\newcommand{\calN}{\mathcal{N}}
\newcommand{\calE}{\mathcal{E}}

\newcommand{\last}{\mathsf{Last}}
\newcommand{\Pref}{\mathsf{Pref}}

\newcommand{\Inf}{\mathsf{Inf}}

\newcommand{\nat}{\mathbb{N}}

\newcommand{\br}{\mathsf{BRTree}}

\newcommand{\update}{\mathsf{Update}}

\newcommand{\truev}{\mathsf{true}}
\newcommand{\falsev}{\mathsf{false}}

\acrodef{dba}[DBA]{deterministic B\"uchi automaton} \title{\LARGE \bf
  Integrating active sensing into reactive synthesis with temporal logic
  constraints under partial observations}

\author{Jie Fu$^{1}$ and Ufuk Topcu$^{1}$
  \thanks{This work is supported by AFOSR grant number
    FA9550-12-1-0302, ONR grant number N000141310778 and NSF CNS award
    number 1446479.}
  \thanks{$^{1}$Jie Fu and Ufuk Topcu are with the Department of
    Electrical and Systems Engineering, University of Pennsylvania,
    Philadelphia, PA, 19104, USA {\tt\small jief,
      utopcu@seas.upenn.edu}.}%
}

\begin{document}

\maketitle

\begin{abstract}
  We introduce the notion of online reactive planning with sensing
  actions for systems with temporal logic constraints in partially
  observable and dynamic environments. With incomplete information on
  the dynamic environment, reactive controller synthesis amounts to
  solving a two-player game with partial observations, which has impractically
  computational complexity. To alleviate the high computational
  burden, online replanning via sensing actions avoids solving the
  strategy in the reactive system under partial observations. Instead,
  we only solve for a strategy that ensures a given temporal logic
  specification can be satisfied had the system have complete
  observations of its environment. Such a strategy is then transformed
  into one which makes control decisions based on the observed
  sequence of states (of the interacting system and its
  environment). When the system encounters a belief---a set including
  all possible hypotheses the system has for the current state---for
  which the observation-based strategy is undefined, a sequence of
  sensing actions are triggered, chosen by an active sensing strategy,
  to reduce the uncertainty in the system's belief. We show that by
  alternating between the observation-based strategy and the active
  sensing strategy, under a mild technical assumption of the set of
  sensors in the system, the given temporal logic specification can be
  satisfied with probability 1.
\end{abstract}

\section{Introduction}
\label{sec:intro}

Control synthesis under partial observations has been an important
topic since complete and precise information (about the system and
environment states) during the execution of a controller is often not
available in practice. However, synthesis methods for systems under
partial observations are of high complexity and have limitations in
their applications. With incomplete information, the problem of
synthesizing a controller in a partially observable Markov decision
process (POMDP) has been shown to be PSPACE-complete, even for finite
planning horizons \cite{littman1996algorithms}. When the control
specification is given in temporal logic and the environment is
dynamic and possibly adversarial, the interaction between a system and
its environment can be captured in a two-player partially observable
game with infinite stages, for which the qualititive-analysis problem
under finite-memory strategies is EXPTIME-complete
\cite{chatterjee2010complexity}.

For temporal logic constraints,
synthesis algorithms for stochastic systems modeled as POMDPs have been
studied in \cite{wongpiromsarn2012control,Rangoli2014}. To deal with a
partially observable, dynamic environment, synthesis algorithms for
two-player game with partial observations have been developed under
two qualitative correctness criteria
\cite{Chatterjee2012,Arnold20037}: \emph{sure-winning} and
\emph{almost-sure winning} controllers.  A sure-winning controller
ensures the satisfaction of a specification whereas an almost-sure
winning controller is a randomized strategy and ensures satisfaction
with probability 1.  These solutions rely on a subset construction and
has complexity exponential in the size of the state space
\cite{chatterjee2007algorithms,chatterjee2010complexity}.

An interesting question that has not been investigated much is the
following: Since the high computational complexity is caused by
incomplete information, is it possible to reduce the computational
effort and still ensure correctness of the control design by acquiring
new information at run time? In this paper, we give a method that
provides a partial, affirmative answer to this question.
Particularly, we study a system with actions to obtain information,
referred to as \emph{sensing actions}, and show how to utilize these
actions in a way that a given \ac{ltl} specification is satisfied
almost surely with reduced computational effort.

The new approach in this paper is inspired by \cite{Brafman2012},
where the authors propose a method of online planning with partial
observations and sensing actions as a way to overcome such complexity
since the system only needs to compute a strategy for a finite number
of steps, and replans with new information obtained through sensing
actions. For temporal logic specifications, online planning method in
\cite{Brafman2012} has no correctness guarantee. We propose a similar
framework of active sensing and reactive synthesis under temporal
logic constraints. The basic approach is the following: During control
execution, the system maintains a \emph{belief}, which is a set of
states it thinks the current state must be in based on its partial
observation for the game history. The belief is updated under two
cases: In one of these cases, the system or the environment makes a
move, the belief is updated to the set of states possibly arrived at
as a result of move. Alternatively, the system can activate a sensor,
detecting the value of some propositional formula and revises its
belief according to the additional information obtained through
sensing. In the second case, the system applies an active sensing
strategy.  A sequence of sensor queries are made to obtain the most
useful information for reducing the system's uncertainty in the
current state. The benefit of performing the combined active sensing
and reactive planning is that we can indeed avoid solving a two-player
zero-sum game with partial observations. Rather, we transform the
sure-winning strategy for the system in the same game with
\emph{perfect observations}, into a \emph{randomized, belief-based}
strategy. By construction, the randomized strategy may not be defined
for every belief the system can encounter at run time.  During control
execution, the system alternates between the randomized strategy and
the active sensing strategy. We prove that if the set of available
sensors meets a sufficient condition, the temporal logic specification
can be satisfied with probability 1, i.e., almost surely.

The rest of the paper is organized as follows. We begin with some
preliminaries and the formulation of the problem in
section~\ref{sec:prelim}. Section~\ref{sec:mainresults} presents the
main results on synthesizing provably correct, online reactive
controllers with sensing actions for temporal logic constraints. In
Section~\ref{sec:example} we illustrate the method using a robot
motion planning example in a partially observed environment.
 \normalcolor

\section{Problem formulation and preliminaries}
\label{sec:prelim}
A probability distribution on a finite set $S$ is a function $D : S
\rightarrow [0,1]$ such that $\sum_{s\in S} D(s)=1$. The set of
probability distributions on a finite set $S$ is denoted
$\mathcal{D}(S)$. The support of $D$ is the set $Supp(D)=\{s\in S\mid
D(s) >0\}$. Let $\Sigma$ be a finite alphabet. $\Sigma^\ast$,
$\Sigma^\omega$, and $\Sigma^+$ are sets of strings over $\Sigma$ with
finite length, infinite length, and length greater than or equal $1$,
respectively. Given $u$ and $v$ in $ \Sigma^\ast$, $uv$ is the
concatenation of $u$ with $v$. A string $u \in \Sigma^\ast$ is a
\emph{prefix} of $w\in \Sigma^\ast$ (or $w\in \Sigma^\omega$) if there
exists $v\in \Sigma^\ast$ (or $v\in \Sigma^\omega$) such that
$w=uv$. For a string $w$, the set of symbols occurring infinitely
often in $w$ is denoted $\Inf(w)$. The last symbol in a finite string
$w$ is denoted $\last(w)$.

\subsection{Game, specification and strategies}

Through abstraction for systems with continuous and discrete dynamics,
the interaction of a system and its dynamic environment can be
captured by a labeled finite-state transition system
\cite{Kloetzer2008,ram-hadas}: \[M=\langle S, \Sigma, \delta, s_0,
\calAP, L\rangle\] where \begin{inparaenum}
\item $S = S_1 \cup S_2$ is the set of states. At each state in  $S_1$,
  the system takes an action. At each state in $S_2$, the environment
  takes an action. 
\item $\Sigma=\Sigma_1\cup \Sigma_2$ is the set of
  actions. $\Sigma_1$ is the set of actions for the
  system,
  and $\Sigma_2$ is the set of actions for the environment.
\item $s_0$ is the initial state.
\item $\delta: S\times \Sigma \rightarrow S$ is the transition
  function.
\item $L:S\rightarrow 2^{\calAP}$ is the labeling function that maps a
  state $s\in S$ to a set of atomic propositions $L(s)\subseteq
  \calAP$ that evaluate true at $s$.
\end{inparaenum}

We use a fragment of \ac{ltl} \cite{alur2001Fragment} to specify the
desired system properties such as safety, reachability, liveness and
stability.  Given a temporal logic formula $\varphi$ in this class,
one can always represent it by a \ac{dba} $\calA_\varphi =\langle H,
2^{\calAP}, \delta_\varphi, h_0, F_\varphi \rangle$ where $H$ is the
set of states, $2^{\calAP}$ is the set of alphabet, $\delta_\varphi:
H\times 2^\calAP\rightarrow H$ is the transition function. $h_0$ is
the initial state and $F_\varphi$ is the set of final states. A word
$w= a_0 a_1\ldots \in (2^\calAP)^\omega$ induces a state sequence
$h_0h_1\ldots \in H^\omega$ where $h_{i+1}=\delta_\varphi(h_i,a_i)$,
for all $i\ge 0$. A word $w$ is accepted in $\calA_\varphi$ if and
only if the state sequence $\rho \in H^\omega$ induced from $w$ visits
some states in $F_\varphi$ infinitely often.

A product operation is applied to incorporate the temporal logic
specification into the labeled transition system, giving rise to a
two-player turn-based B\"uchi game between the system (player 1) and
its environment (player 2):
\[
G= \langle Q, \Sigma, T, q_0, F \rangle =M\ltimes \calA_\varphi \]
where the components are defined as follows.
\begin{itemize}
\item $Q=Q_1\cup Q_2$ is the set of states, where $Q_1=S_1\times H$
  and $Q_2= S_2\times H$.
\item $T: Q\times \Sigma \rightarrow Q$ is the
  transition function. Given $(s,h)\in Q$, $\sigma \in \Sigma$, if
  $\delta(s, \sigma)=s'$, then
  $T(q,\sigma)=q'$ where $q'=(s', \delta_\varphi(h, L(s')))$.
\item $q_0=(s_0,\delta_\varphi(h_0,L(s_0)))$ is the initial state.
\item $F\subseteq Q\times F_\varphi$ is a subset of states that
  determines a B\"uchi winning condition.
\end{itemize} 

A \emph{play} in $G$ is either a finite sequence of interleaving
states and actions $\rho =q_0a_0q_1a_1\ldots q_n \in (Q\cup
\Sigma)^\ast Q$ or an infinite sequence $\rho=q_0a_0q_1a_1\ldots \in
(Q\cup\Sigma)^\omega$ such that $q_0$ is the initial state and
$T(q_i,a_i)=q_{i+1}$ for all $i\ge 0$. If $\rho$ is finite, the last
element of $\rho$ is a state, denoted $\last(\rho)$.  An infinite
\emph{play} $\rho$ is \emph{winning} for player 1 in $G$ if and only
if $\Inf(\rho)\cap F\ne \emptyset$.

In game $G$, each state in $Q$ is associated with a truth assignment
to a set $\calP$ of predicates. Note that $\calP $ may not equal
$\calAP$. This association is captured by the \emph{interpretation}
function $\pi$ such that for any $q\in Q$, for any predicate $ p \in
\calP$, $\pi(q)(p) \in \{\truev,\falsev\}$. We write $\pi(q) =\land_{p
  \in \calP} \ell_p$ where $\ell_p =p$ if $\pi(q)(p)=\truev$ and
$\ell_p =\neg p$ if $\pi(q)(p)=\falsev$, $\land$, $\neg$ are the
logical connectives for conjunction and negation, respectively. In the
set $\calP$, there is a predicate $t$ indicating whose turn it is to
play: If $t=1$, then the system takes an action, otherwise the
environment makes a move. It is assumed that the value of $t$ is
globally observable, which means, the system always knows whose turn
it is to play.

We consider the case when the system has partial observation of values
for the set $\calP$ of predicates.  Following \cite{Chatterjee2012},
this partial observation can be defined by an equivalence relation
over the set of states, denoted $\mathcal{R} \subseteq Q\times Q$. Two
states $q$ and $q'$ are \emph{observation-equivalent}, that is,
$(q,q')\in \mathcal{R}$, if both $q$ and $q'$ provide the same state
information observable by the system, i.e., the value of $p\in \calP $
is observable at $q$ if and only if it is observable at $q'$, and
$\pi(q)(p)=\pi(q')(p)$.  We denote the \emph{observations of states}
for the system by $\calO\subseteq 2^{Q}$, which is defined by the
observation-equivalence classes. Clearly, $\calO$ is a partition of
the state space. We define an \emph{observation function} $\obs: Q
\cup \Sigma \rightarrow \calO \cup \Sigma_1\cup \{-\}$ such
that \begin{inparaenum}[1)]\item $q\in \obs(q)$; \item for every
  $q_1$, $q_2\in \obs(q)$, $(q_1,q_2) \in \mathcal{R}$, \item if
  $\sigma \in \Sigma_1$ $\obs(\sigma)=\sigma$; and \item if $\sigma
  \in \Sigma_2$, $\obs(\sigma)= - $
\end{inparaenum}. The last two properties express that the system
observes (knows) which action it performed but does not directly
observe the action of the environment. The information received by the
system on the environment's action is from the effect of that action,
reflected in the observed arrived state.

The observation sequence of a play $\rho= q_0 a_0 q_1\ldots$ is a
sequence $\obs(\rho)=\obs(q_0)\obs(a_0)\obs(q_1)\ldots$.  It is worth
mentioning that two states $q=(s,h)$ and $q=(s',h')$ can be
observation-equivalent even if $h\ne h'$. Therefore, two
observation-equivalent $\rho$ and $\rho'$ can differ in their state
projections onto the set $Q$ of states in the specification automaton
$\calA_\varphi$.

Let $\Pref(G)$ denote the set of finite prefixes of all plays in $G$,
each of which ends with a state in $Q$.  For both players $1$ and $2$,
a \emph{deterministic} strategy for player $i$ is a function $f_i:
\Pref(G) \rightarrow \Sigma_i $ and a \emph{randomized} strategy is a
function $f_i: \Pref(G) \rightarrow \mathcal{D}(\Sigma_i) $.  We say
that player $i$ \emph{follows} strategy $f_i$ if for any finite prefix
$\rho\in \Pref(G)$ at which $f_i $ is defined, player $i$ takes the
action $f_i(\rho)$ if $f_i$ is deterministic, or an action $\sigma \in
Supp(f_i(\rho))$ with probability $f_i(\rho)(\sigma)$ if $f_i$ is
randomized. Since the system has partial information of the states, it
can only execute an \emph{observation-based} strategy $f_1$, in the
sense that if for any two prefixes $\rho$ and $\rho' \in \Pref(G)$, if
$\obs(\rho) = \obs( \rho')$, then $f_1(\rho)=f_1(\rho')$.  A strategy
is \emph{memoryless} if and only if $f_i(\rho)=f_i(\last(\rho))$.  For
B\"uchi game $G$ with complete information, there exists a
\emph{deterministic, memoryless} winning strategy for one of the
players.

\subsection{Partial observation, belief and sensing actions}

With partial observations, the system keeps track of the play in the
game by maintaining and updating a set $B\subseteq Q$ of states,
referred to as the \emph{belief}, which is the set of states the
system thinks the game can be in, given the observation history. In
which follows, we show how the belief is obtained and updated. The set
of beliefs in the game is denoted $\calB\subseteq 2^Q$.  We define a
function $\alpha: \Pref(G) \rightarrow \calB$ that maps a prefix of
$g$ into a \emph{belief} as follows: given a prefix $\rho =
q_0a_0\ldots q_n$, the belief of the system is $\alpha(\rho) =
\{\last(\rho') \in Q \mid \rho' \in \Pref(G) \text{ and }
\obs(\rho')=\obs( \rho) \}$.

During the interaction with the environment, the system's belief is
updated in two ways:
\begin{inparaenum}[(i)]
\item The system applies a control action, obtains a
new observation of the arrived state, and updates its belief to
one in which the current state could be. 
\item The environment takes some action. The system obtains an
  observation $o\in \cal O$ of the arrived state, and subsequently
  updates its belief that includes its hypothesis for the current
  state. \end{inparaenum} Formally, this process is called
\emph{belief update}, which can be captured by the function
\begin{equation}
\label{beliefupdate}
\update: \calB\times (\Sigma_1 \cup \{-\} )\times \calO \rightarrow
\calB,
\end{equation}
It is reminded that the symbol ``$-$'' is the observation for an
action of the environment. Given a belief $B$, the system takes an
action $a\in \Sigma_1$ and gets an observation $o\in \calO$.  Then it
updates its belief to $B' = o \ \cap \update(B, a, o)= \{q' \mid \exists
q \in B \text{ such that } T(q,a)= q'\}$. If it is the environment's
turn, after the environment takes some action, the system gets an
observation $o\in \calO$ and then updates its current belief $B$ to
$B' =\update(B, -, o) = o\cap \{ q' \mid \exists q \in B,\exists
\sigma \in \Sigma_2 \text{ such that } T(q,\sigma)=q' \} $.

We distinguish a set $\Gamma$ of \emph{sensing actions} for the system
and explain how the sensing actions affects the system's belief as
follows.
\begin{definition}
\label{def:sensingact}
Consider the set $\calP$ of atomic propositions and the set $\Gamma$
of sensing actions. For each sensing action $a\in \Gamma$, there
exists at least one propositional formula $\phi$ over $\calP$ such
that after applying the sensing action $a$, the truth value of $\phi$
is known.  Depending on the value of $\phi$, the system can partition
a belief $B$ into two subsets, expressed by
\[
\knows(\phi,a,B)\defeq (B', B\setminus B'),
\]
where $B'$ is the set of states in which $\phi$ evaluates true and
$B\setminus B'$ is the set of states in which $\phi$ evaluates false.
Hence, if $\phi$ is true, the belief is revised to be $B'$,
otherwise to be $B\setminus B'$.
\end{definition}
To capture both global and local sensing capabilities, for a
given state $q$, we denote $\Gamma_q \subseteq \Gamma$ to be a set of
sensing actions \emph{enabled} at $q$. The set of sensing actions
enabled at a belief $B\subseteq Q$ is $\bigcap_{q\in B} \Gamma_q$.

The following assumption is made for sensing actions.
\begin{assumption}
\label{nosideeffect}
A sensing action will not change the value of variables and/or
predicates in $ \calP$.
\end{assumption}
The assumption is not restrictive because if an action introduces both
physical and epistemic changes, we simply consider it as an ordinary control action
and include it into $\Sigma_1$. We call an action in $\Gamma$
\emph{sensing} to emphasize that it provides information of the
current state, and an action in $\Sigma$ \emph{physical} to emphasize
it changes the state of the game.  We assume that at each turn of the
system, it can either choose a physical action, or several sensing
actions followed by a physical action.


We solve the following problem in this paper.
\begin{problem}
\label{problem}
Given a two-player turn-based B\"uchi game $ G= \langle Q, \Sigma, T,
q_0, F \rangle $, and a set $\Gamma$ of sensing actions,
design an observation-based strategy $f: Q^\ast \rightarrow
D(\Sigma_1) \cup \Gamma^\ast$ with which the specification is
satisfied with probability $1$, i.e., almost surely, whenever such a
strategy exists.
\end{problem}

\section{Main results}
\label{sec:mainresults}
For games with partial information, algorithms in
\cite{chatterjee2007algorithms} can be used to synthesize
observation-based controllers which ensure given temporal logic
specifications are satisfied surely, or almost surely, i.e., with
probability 1, whenever such controllers exist. In this paper, we
only consider the cases in which observation-based controllers do not
exist and thus require additional information at run time for
satisfying given temporal logic specifications. We distinguish two
phases in the online planning: \emph{Progress} phase and
\emph{sensing} phase. As the names suggest, during the progress phase,
the system takes physical actions in order to satisfy the temporal
logic constraints, and during the sensing phase, the system takes
sensing actions to reduce the uncertainty in its belief for the
current game state. The transition from one phase to another will be
explained after we introduce the methods for synthesizing strategies
used in both phases.
\subsection{A belief-based strategy for making progress}
For a game with partial observation, 
we aim to synthesize a belief-based, memoryless and randomized
strategy $f_P : \calB \rightarrow \calD( \Sigma_1)$ that can be
applied for making progress towards satisfying the given \ac{ltl}
fragment formula $\varphi$. 

In the two-player B\"uchi game $G$, the deterministic sure-winning
strategy $\mathsf{WS}: Q \rightarrow \Sigma_1$ can be computed (with
methods in \cite{Gradel2002}) but requires complete information to
execute at run time. The belief-based strategy $f_P$ is constructed
from the sure-winning strategy $\mathsf{WS}$ in the following way: Let
$\mathsf{Win}_1\subseteq Q$ be the set of states at which
$\mathsf{WS}$ are defined. Given $B \in \calB$, let
\begin{multline*}
\mathsf{Progress}(B) = \bigcup_{q\in B} \mathsf{WS}(q), \text{ and }
\\
\mathsf{allow}(B)= \bigcap_{q\in B} \mathsf{allow}(q),\\
\text{ where } \mathsf{allow}(q)=\{\sigma\in \Sigma_1\mid T(q, \sigma)\in
\mathsf{Win}_1\}.
\end{multline*}
For each state $q\in B$, the sure-winning strategy will suggest action
$\mathsf{WS}(q)$ to be taken by the system, which is then included
into a set $\mathsf{Progress}(B)$. The set $\mathsf{allow}(B)$ is a set of actions
with the following property: No matter in which state of $B$ the game
is, by taking an action in $\mathsf{allow}(B)$, the next state will still be
one for which the sure-winning strategy is defined. Then, if $
\mathsf{Progress}(B) \subseteq \mathsf{allow}(B)$, we let $f_P(B)(\sigma) =
\frac{1}{\abs{\mathsf{Progress}(B)}} $ for each $\sigma \in
\mathsf{Progress}(B)$. Otherwise, $f_P$ is undefined for $B$. Note that since
the computation $f_P$ can be essentially reduced to computing the
interaction of two sets, there is no need to compute $f_P$ for all
possible subset of $Q$. Rather, we can efficiently compute $f_P$ for
each belief $B$ encountered at run time.

We have transformed the sure-winning strategy with complete
information in the B\"uchi game into a randomized, belief-based
strategy.  During control execution, the system maintains its current
belief. At each turn of the system, after applying an action $\sigma
\in \Sigma_1$ at the state $B$, the system receives an observation
$o\in \cal O$, updates its belief to $B'= \update(B, \sigma, o)$. When
it is a move made by the environment, the system obtains another
observation $o'\in \calO$, updates its belief to $ B''= \update(B', -,
o')$. The system applies $f_P(B'')$ as long as $f_P$ is defined for
$B''$. When $f_P$ is undefined for the current belief $B$, then we
switch to the sensing phase for actively acquiring more information to
reduce the uncertainty in its current belief.

\subsection{An active sensing strategy for reducing uncertainty}
\label{sec:activesense}

During the progress phase with the randomized, belief-based strategy
$f_P$, if the system runs into a belief at which $f_P$ is undefined,
it needs to update its belief through sensing until either it finds
itself in a state for which $f_P$ is defined, or it cannot further
refine its belief: A belief $B$ cannot be refined if for any sensing
action $a$ enabled at $B$ and for any formula $\phi$ such that
$(B_1,B_2) = \knows(\phi, a, B)$, it holds that for either $i=1$ or
$i=2$, $B_i = B$.  We represent the process of belief revision with
sensing actions as a tree structure, referred to as a \emph{belief
  revision tree}, and then propose a synthesis method for an active
sensing strategy using the belief revision tree.

Given a belief $B^o \in \calB$, the \emph{belief revision tree} with
the root $B^o$ is a tuple $\br(B^o) = \langle \calN, \calE \rangle$,
where $\calN$ is the set of nodes in the tree, consisting a subset of
beliefs, and $\calE \subseteq \calN \times \Gamma\times \calN$ is the
set of edges. It is constructed as follows.
\begin{enumerate}
\item The root of the tree is $B^o$.
\item At each node $B \in \calN$, for each enabled sensing action
  $a\in \Gamma_B$, if there exists a formula $\phi$ such that
  $(B_{1},B_{2})=\knows(\phi, a, B)$ and both $B_1, B_2 $ are not
  empty, then we add two children $B_1,B_2$ of $B$, and include edges
  $(B, a, B_1)$, $(B, a, B_2)$ into the edges $\cal E$.
\item A node $B$ is a leaf of the tree if and only if
  either \begin{inparaenum}[1)]\item $B$ cannot be further revised by
    any sensing action, or \item $f_P$ is defined for $B$.
\end{inparaenum}
\end{enumerate}


The active sensing strategy $f_S: \calB \rightarrow \Gamma$ is
computed as follows. First, in the tree $\br(B^o)$, we compute a set
of target nodes $\mathsf{Reach} \subset \calN$ such that a node $B'$ is
included in $ \mathsf{Reach}$ if and only if $f_P(B')$ is defined. The
objective is to apply the least number of sensing actions in order to
reach a belief in $\mathsf{Reach}$ for which $f_P$ is defined. For this
purpose, we have the following recursion:
\begin{enumerate}
\item $X_0=\mathsf{Reach}$, $i=0$.
\item $X_{i+1}= X_i\cup \{ B\in \calB \mid \exists a\in \Gamma, \text{
    such that } \forall B' \in \calB, (B, a, B')\in \calE, B'\in X_i
  \}$ and let $f_S(B)=a$. In other words, a belief $B$ is
  included into $X_{i+1}$ if there exists a sensing action $a$ such
  that when $a$ is applied at  $B$, no matter which
  belief the system might reach, it must be in $X_i$.
\item Until $i$ is increased to some number $m\in \nat$ such that $
  X_{m+1}=X_m$, we output the sensing strategy $f_S$ obtained so
  far.
\end{enumerate}
We denote $X_m= \mathsf{attr}(\mathsf{Reach}) $, following the notion of an
\emph{attractor} of the set $\mathsf{Reach}$. For any state in
$\mathsf{attr}(\mathsf{Reach})$, there exists a sensing strategy $f_S$ such that for
\emph{whatever outcome} resulted by applying sensing actions, the
system can arrive at some belief in $\mathsf{Reach}$ in finitely many steps by
following $f_S$. Furthermore, it can be proven that $f_S$ minimizes
the number of sensing actions required for the sensing phase under the
constraint that the system will not run into a dead end, which is a
belief that cannot be further refined yet is undefined by $f_P$. The
number of sensing actions during the sensing phase is upper bounded by
the index $i$ for which $B^o\in X_i$ and $B^o\notin X_{i-1}$. The
proof follows from the property of attractor \cite{Gradel2002} and is
omitted here.

\paragraph*{Remark}It is worth mentioning that for a given belief $B$,
the active sensing strategy is unique. Thus, we can store and
continuously update a set of active sensing strategies synthesized at
run time: When the system encounters a belief $B$ for which $f_P$ is
undefined but it has seen before, it can use the stored active sensing
strategy for $B$ without recomputing a new one.  For a large-scale
system with a large number of sensing actions, one can also
pre-compute a library of active sensing strategies and then augment
the library with new active sensing strategies computed at run time.

\subsection{A composite, almost-sure winning strategy}
\label{subsec:almostsurewin}
At run time, the system alternates between strategy $f_P$ for making
progress and strategy $f_S$ for refining its belief. We name the
system's strategy at run time a \emph{composite} strategy, denoted
$f:\calB \rightarrow D(\Sigma_1) \cup \Gamma$, defined by,
\begin{equation}
\label{eq:composite}
f(B) = \left\{ \begin{matrix}
f_P(B) & \text{ if } f_P(B) \text{ is defined.}\\
f_S(B) & \text{ if } f_S(B) \text{ is defined.}\\\end{matrix}\right.
\end{equation}
Note that by construction, the domains of $f_P$
and $f_S$ is always disjoint.

The following assumption provides a sufficient condition for avoiding dead-ends
at run time. 
\begin{assumption}
\label{assume1}
For each state $B$ encountered during the progress phase, if $f_P(B)$
is undefined, then $f_S(B)$ is defined.
\end{assumption}

Since we cannot predict which beliefs the system might have during
control execution with online planning, in the extreme case, for each
predicate $p\in \calP$, we need to have a sensing action or a
combination of sensing actions to detect its truth value. However,
this condition is not necessary and may include some sensing actions
that will never be used at run time.  As the system does not need to
know the \emph{exact} state by extensive sensing, it is at the
system's disposal whether to apply a sensing action and what shall be
applied.

Next we prove the correctness of the composite strategy. To this end,
we recall some property in the solution for B\"uchi games with
complete information from \cite{Gradel2002}: The winning region of the
B\"uchi game $G$ can be partitioned as $\mathsf{Win}_1=
\bigcup_{i=0}^m W_i$ for some $m \in \nat$, $m \ge 0$. For any state
$q\in \mathsf{Win}_1$, there exists a unique ordinal $i$ such that
$q\in W_i$. If $q\in Q_1\cap W_i$ for some $0<i \le m$, then the
winning strategy on $q$ outputs $\sigma \in \Sigma_1$, with which the
system reaches a state $q' \in W_{i-1}\cap Q_2$. If $i=0$, then with
the action $\mathsf{WS}(q)$, we arrive at a state $q' \in
\mathsf{Win}_1$. If $q\in Q_2$, then for any action $\sigma \in
\Sigma_2$ enabled at $q$, $T(q,\sigma) \in W_{i-1}$ if $i \ne 0$, or
$q'\in \mathsf{Win}_1$ otherwise.

\begin{lemma}
  Given a game $G=\langle Q, \Sigma, T, q_0, F \rangle$. Let $B_0 =
  \obs(q_0)$ be the initial belief. If Assumption~\ref{assume1} is
  satisfied and $q_0 \in \mathsf{Win}_1$, the composite strategy $f$
  defined by \eqref{eq:composite} ensures that some states in $F$ of
  $G$ is infinitely often visited with probability 1.
\end{lemma}
\begin{proof}Consider an arbitrary belief $B \in \calB$ for which
  $f_P$ is defined.  By definition of $f_P$, for each $\sigma \in
  \mathsf{Progress}(B)$, the probability of choosing action $\sigma$ is
  $\frac{1}{u}$, where $u=\abs{\mathsf{Progress}(B)}$.  If the actual state is
  $q$ and $q \in W_i$, for some $i\ne 0$, then with probability
  $\frac{1}{u}$, the system will reach a state in $W_{i-1}$.  Thus,
  the probability of the next state being in $W_{i-1}$ is $\frac{1}{u}
  \ge \frac{1}{\abs{Q}} >0$. For other $\sigma'\in f_P(s)$, $\sigma '
  \ne \mathsf{WS}(q)$, the next state after taking $\sigma'$ is in
  $W_j$ for some $0\le j \le m$.  Let $\Pr(q, \lozenge^i W_0)$ denote
  the probability of reaching $W_0$ from state $q$ in $i$ turns. When
  system applies the strategy $f$, it is $\Pr(q,\lozenge^i W_0) \ge
  (\frac{1}{\abs{Q}})^i >0$ and the probability of \emph{not} reaching
  $W_0$ in $i$ turns is \emph{less than or equal to} $
  1-(\frac{1}{\abs{Q}})^{i} \le 1- (\frac{1}{\abs{Q}})^{m+1} =r <1 $
  where $m+1$ is the total number of partitions in
  $\mathsf{Win}_1$. If after $i$ steps the state is not in $W_0$, it
  must be in $W_j$ for some $ 0 < j\le m$, and again the probability
  of not reaching $W_0$ in $m$ steps is less than or equal to $r$.
  Therefore, under the policy $f$, the probability \emph{eventually}
  reaching $W_0$ from any state $q \in \mathsf{Win}_1$ is
  $\mathrm{Pr}(v,\lozenge W_0) =\lim_{k\rightarrow \infty}
  \mathrm{Pr}(v,\lozenge^k W_0) =\lim_{k\rightarrow \infty}( 1-
  \mathrm{Pr}(v, \neg \lozenge^k W_0)) =\lim_{k\rightarrow \infty} (
  1-r^{k/m}) =1-\lim_{k\rightarrow \infty} r^{k/m} =1$.

  Once entering $W_0$, the system will take an action to remain in
  $\mathsf{Win}_1$, and the above reasoning applies again. In this
  way, in the absence of dead ends (Assumption~\ref{assume1}), the
  system can revisit the set $W_0 $ of states with probability 1 by
  following the composite strategy $f$.  Since $W_0\subseteq F$, the
  probability of system always eventually visiting some states in $F$
  is 1.
\end{proof}

To conclude this section, Algorithm~\ref{alg:sensingact} describes the
procedure of online planning with sensing actions.

\vspace{-2ex}
\begin{figure}[hb]
\centering
\includegraphics[width=0.45\textwidth]{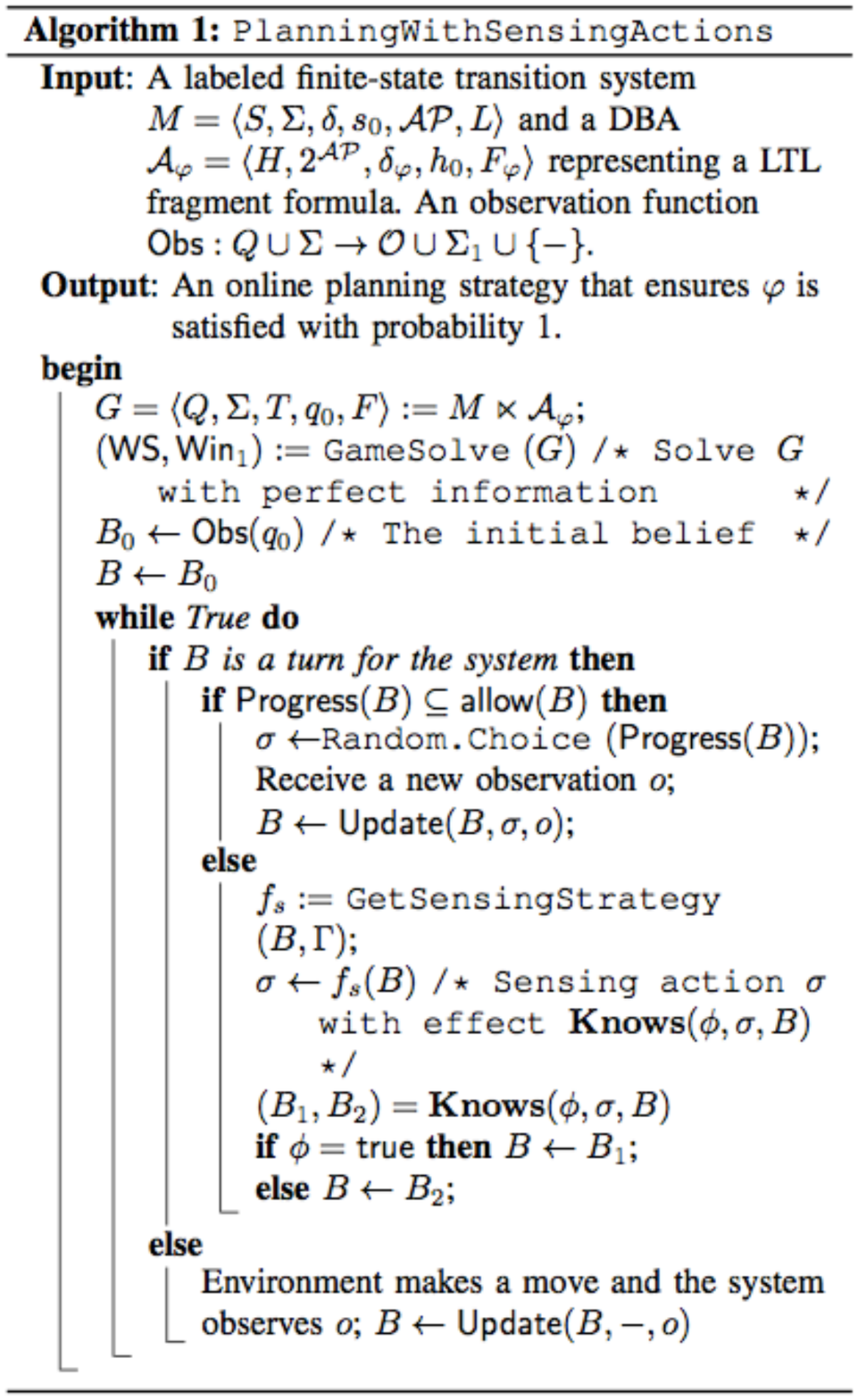}
\caption{Algorithm: PlanningWithSensingActions}
\label{alg:sensingact}
\end{figure}


\section{Examples}
\label{sec:example}
We apply the algorithm to a robotic motion planning example, which is
a variant of the so-called ``Wumpus game'' in a $7\times 7$ gridworld.
Figure~\ref{fig:wumpusgame} consists of one mobile robot, one monster
called ``Wumpus''. The robot is capable of moving in eight compass
directions with actions `N', `S' , `E', `W', `NE', `NW', `SE', `SW'
(horizontally, vertically and diagonally), one step at a time. The
robot and the Wumpus does not move concurrently. The Wumpus can move
in four compass directions with actions `N', `S', `E' and `W' within a
restricted area $\mathsf{Region}$ and emits stench to its surrounding
cells. The objective of the robot is to infinitely revisit region
$R_1$, $R_2$, and $R_3$ in this order, while avoiding running into the
Wumpus. Formally, the temporal logic formula is $\varphi= \square
\lozenge (x_r,y_r) = R_1 \land \lozenge \left((x_r,y_r) = R_2 \land
  \lozenge (x_r,y_r) = R_3\right)\land \square \neg (x_r=x_w\land
y_r=y_w)$ where $(x_r,y_r), (x_w,y_w)$ are the positions of the robot
and the Wumpus, respectively.  Yet, the robot only knows his own
position. For this case of partial observation, without the inclusion
of sensing actions, it can be shown that with the algorithms in
\cite{chatterjee2007algorithms}, observation-based, sure-winning
strategies and almost-sure winning strategies do not exist.

\begin{figure}[ht]
\centering
\includegraphics[width=0.25\textwidth]{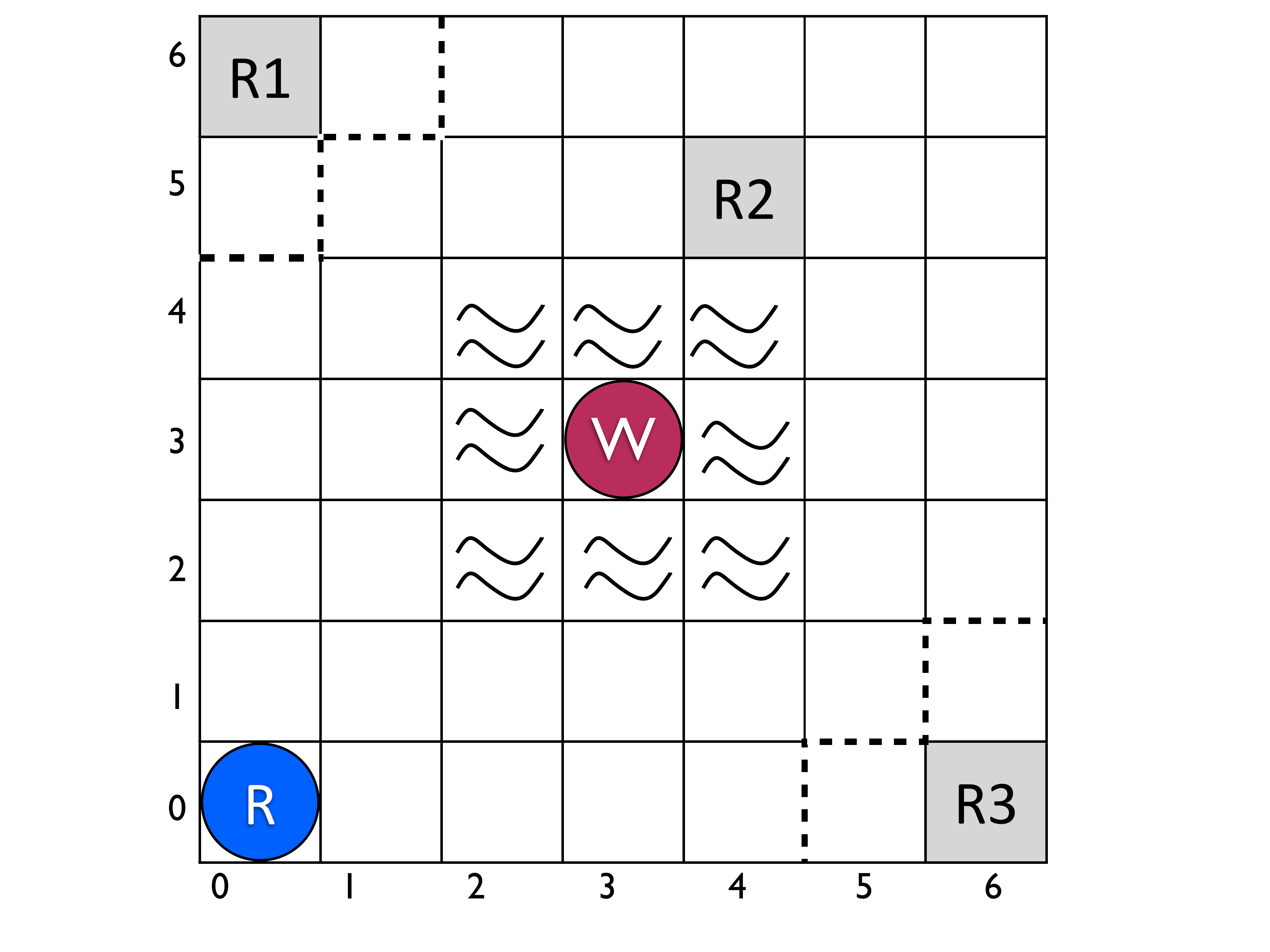}
\caption{The gridworld with a robot (R) and the Wumpus (W). The grey
  cells are regions $R1$, $R2$ and $R3$. The Wumpus is restricted to
  the area inside the dash lines. The stenches emitted by the Wumpus
  are represented by waves.}
\label{fig:wumpusgame}
\end{figure}

Here, we introduce a set of sensing actions to the game. For the robot
to know the position of the moving obstacles, it needs to apply a
sensing action --- $\mathsf{smell}(x,y)$ to detect if there exists
stench at cell $(x,y)$. Thus, when the robot applies
$\mathsf{smell}(x,y)$, if the result is True, then the Wumpus must be
some cells in the set $S=\{(x',y')\mid x' \le x+1, y'\le y+1, x',y'\in
\nat\} \cap \mathsf{Region}$. Otherwise, it is not possible that the
Wumpus is in any cell in $S$.

We illustrate how the robot updates his belief using sensing action
$\mathsf{smell}(x,y)$ where $(x,y)$ is a cell in the
gridworld. Suppose that the robot does not know where the Wumpus is
and hypothesizes it can be in any cell in the $\mathsf{Region}$. Once
it applies the sensing action $(2,2)$, since the cell has stench and
the sensor returns True. Then, immediately the robot will know the
Wumpus is in one of the cells in the set $S=
\{(1,1),(2,1),(3,1),(1,2),(2,2), (3,2), (3,1), (3,2), (3,3)\}$,
because only if the Wumpus is in a cell of $S$, there can be stench in
cell $(2,2)$.

From the numerical experimental result, after $1000$ steps (a step
includes either a robot's (sensing or physical) action or a movement
of the Wumpus), the robot visited the set $F$ in the formulated
two-player game $G$ $14$ times and can continue to visit $F$ infinite
often. In Figure~\ref{fig:beliefUpdate} we show the belief updates by
applying alternatively the exploitation strategy and active sensing
strategy for the initial $100$ steps. It is observed that the maximum
cardinality of the belief set is $43$ over the control execution,
which means that the robot thinks the Wumpus can be in any cell in its
restricted region. However, if there is no danger of running into the
Wumpus in a few next steps, there is no need to exercising any sensing
action. The implementations are in Python on a desktop with Intel(R)
Core(TM) i5 processor and 16 GB of memory. The average time for the
robot making a decision is $8.55\times 10^{-4}$ seconds. The
computation of the product game took $40.14$ seconds and the winning
strategy under complete information is computed within $14$ seconds.

\begin{figure}[ht]
\centering
\includegraphics[width=0.5\textwidth]{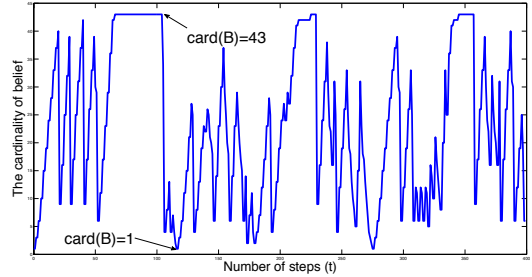}
\caption{The update in the number of possible Wumpus' positions in the
  system's belief. }
\label{fig:beliefUpdate}
\end{figure}

\section{Conclusions}
\label{sec:conclusions}
Our work shows that when additional information can be obtained
through sensing actions, one can transform a sure-winning strategy
with complete information to a belief-based, randomized strategy,
which is then combined, at run time, with an active sensing strategy
to ensure a given temporal logic specification is satisfied with
probability 1. The synthesis method avoids a subset construction for
solving games with partial information. Meanwhile, the active sensing
strategy leads to a cost-efficient way of sensor design: Although we
require a sufficient set of sensing actions to avoid dead-ends at run
time, the system minimizes the usage of sensing actions by asking the
most revealing queries, depending on what specification is to be
satisfied, and how much uncertainty the system has about the game
state at run time. In future work, we will consider more examples for
practical robotic motion planning under partial observations. It is
also important to consider the uncertainty in the sensors. For
example, a sensor query might return a probabilistic distribution over
a set of states, rather than a binary answer to proposition logical
formulae considered herein. For this extension, we are currently
investigating modifications that need to be made to account for
delays, uncertainty in the information provided by the
sensors. 

\balance

\end{document}